\documentclass[12pt,article]{IEEEtran}
\usepackage{epsfig}
\usepackage{amsmath}
\usepackage{amstext}
\usepackage{amsfonts}
\usepackage{amssymb}
\usepackage{eucal}
\usepackage{graphicx}
\usepackage{url}

\newcommand{\RR}{{\mathbb{R}}}

\newcommand{\CC}{{\mathbb{C}}}
\newcommand{\tr}{{\mathrm{tr}}}

\newcommand{\abs}[1]{\left\lvert#1\right\rvert}
\newcommand{\norm}[1]{\left\lVert#1\right\rVert}

\newcommand{\loor}[1]{\left\lfloor#1\right\rfloor}

\newcommand{\Expectation}{{\mathbb{E}}}
\newcommand{\Expect}[2]{{\Expectation_{#1}\left[#2\right]}}
\DeclareMathOperator{\SINR}{SINR}

\newcommand{\yy}{{\bf y}} 
\newcommand{\YY}{{\bf Y}}
\newcommand{\WW}{{\bf W}} 
\newcommand{\PP}{\sqrt{\bf P}} 
\newcommand{\SIGNAL}{{\bf s}} 
\newcommand{\NOISE}{{\bf n}} 
\newcommand{\HC}{{\bf C}}
\newcommand{\HH}{{\bf H}}
\newcommand{\VV}{{\bf V}}
\newcommand{\Vv}{{v}}
\newcommand{\vv}{{\bf v}}
\newcommand{\ww}{{\bf w}}
\newcommand{\powcoef}[1]{\sqrt{P_{#1}}}
\newcommand{\signal}{s}
\newcommand{\FF}{{\bf F}}
\newcommand{\FH}{{\FF^H}}
\newcommand{\DD}{{\bf H}}
\newcommand{\ddd}{{h}}
\newcommand{\uuu}{{w}}
\newcommand{\dd}[2]{{\ddd_{{#2}{#1}}}}
\newcommand{\uu}[2]{{\uuu_{{#2}{#1}}}}
\newcommand{\rrho}{{\varrho}}

\newcommand{\E}{E}
\newcommand{\betastar}{\beta^{\star}}
\newcommand{\betaplus}{\beta^{+}}
\newcommand{\betamf}{\beta^{\textrm{MF}}}
\newcommand{\betammse}{\beta^{\textrm{MMSE}}}
\newcommand{\betasic}{\beta^{\textrm{SIC}}}

\newcommand{\Pmax}{P_{\text{max}}}
\newcommand{\hmax}{h_{\text{max}}}
\newcommand{\PA}{{\mathbb P}}
\newcommand{\p}{{\bf p}}
\newcommand{\Cste}{\Omega}

\newcommand{\R}{{{\bf R}}}
\newcommand{\Gam}{{C}} 

\newcommand{\Gammmse}{\Gam^{\textrm{MMSE}}}
\newcommand{\Gamopt}{\Gam^{\textrm{OPT}}}
\newcommand{\Gami}{{\gamma}}
\newcommand{\Gamoptstar}{\Gam^{\star}}

\DeclareMathOperator{\Trace}{Trace}
\newcommand{\HHu}[1]{{{\bf h}_{#1}}}

\newcommand{\Hu}[1]{{{\HH}_{(-#1)}}}

\newcommand{\WWu}[1]{{{\bf w}_{#1}}}

\newcommand{\Wu}[1]{{{\WW}_{(-#1)}}}

\newcommand{\Pu}[1]{{\PP}_{(-#1)}}
\newcommand{\Gu}[1]{{\GG}_{(-#1)}}
\newcommand{\GuH}[1]{{\GG}_{(-#1)}^H}

\newcommand{\Id}{{\bf I}}
\newcommand{\g}{{\bf g}}
\newcommand{\GG}{{\bf G}}

\newcommand{\ala}{\eta}
\newcommand{\aaa}[2]{{\ala_{#1}({#2})}}
\newcommand{\ttau}[2]{{\tau_{#1}({#2})}}
\newtheorem{lemma}{Lemma}

\newtheorem{proposition}{Proposition}
\newtheorem{theorem}{Theorem}
\newtheorem{definition}{Definition}

\begin{document}
\bibliographystyle{/user/nbonneau/home/latex/IEEEtran}

\title{Non-atomic Games for Multi-User Systems}

\author{\authorblockN{Nicolas Bonneau\authorrefmark{1},
M{\'e}rouane Debbah\authorrefmark{2},
Eitan Altman\authorrefmark{1}, and
Are Hj{\o}rungnes\authorrefmark{3}\\}
\authorblockA{\authorrefmark{1} MAESTRO, INRIA Sophia Antipolis,
2004 Route des Lucioles, B.P. 93,
06902 Sophia Antipolis, France\\
Email: \url{{nicolas.bonneau, eitan.altman}@sophia.inria.fr}\\}
\authorblockA{\authorrefmark{2}Mobile Communications Group, Institut Eurecom,
2229 Route des Cretes, B.P. 193,
06904 Sophia Antipolis, France\\
Email: \url{merouane.debbah@eurecom.fr}\\}
\authorblockA{\authorrefmark{3}UniK--University Graduate Center,  University of Oslo,
Instituttveien 25, P.~O.~Box 70, N-2027 Kjeller, Norway\\
Email: \url{arehj@unik.no}}
}

\maketitle

\begin{abstract}
In this contribution, the performance of a multi-user system is
analyzed in the context of frequency selective fading channels.  Using
game theoretic tools, a useful framework is provided in order to
determine the optimal power allocation when users know only their own
channel (while perfect channel state information is assumed at the
base station). We consider the realistic case of frequency selective
channels for uplink CDMA. This scenario illustrates the case of
decentralized schemes, where limited information on the network is
available at the terminal.  Various receivers are considered, namely
the Matched filter, the MMSE filter and the optimum filter. The goal
of this paper is to derive simple expressions for the non-cooperative
Nash equilibrium as the number of mobiles becomes large and the
spreading length increases. To that end two asymptotic methodologies
are combined. The first is asymptotic random matrix theory which
allows us to obtain explicit expressions of the impact of all other
mobiles on any given tagged mobile.  The second is the theory of
non-atomic games which computes good approximations of the
Nash equilibrium as the number of mobiles grows. \footnote{This work
was supported by the BIONETS project
\url{http://www.bionets.org/} and by the Research Council of Norway
through the OPTIMO project ``Optimized Heterogeneous Multi-user MIMO Networks''. }
\end{abstract}

\section{Introduction}
\label{sec:introduction}

Resource allocation is of major interest in the context of multi-user
systems. In the uplink multi-user systems, it is important for users
to transmit with enough power to achieve their requested quality of
service, but also to minimize the amount of interference caused to
other users. Thus, an efficient power allocation mechanism allows to
prevent an excessive consumption of the limited ressources of the
users.

The most straightforward way to design a power allocation (PA)
mechanism is as a centralized procedure, with the base station
receiving training sequences from the users and signaling back the
optimal power allocation for each user. Power control schemes in
cellular systems were first introduced for TDMA/FDMA
\cite{ieeeVT.zander92,ieeeVT.grandhi93}; more recently an optimal scheme
was derived for Code Division Multiple Access (CDMA)
\cite{ieeeVT.wu99}. In order to achieve the optimal capacity, the users
may also be sorted according to some rule of precedence
\cite{ieeeIT.tse98}. However, this involves a non negligible overhead
and numerous non informational transmissions. In addition, the
complexity of centralized schemes increases drastically with the number of
users. As discussed in \cite{ieeeAP.elosery00}, centralized algorithms
generally do not have a practical use for real systems, but provide
useful bounds on the performance that can be attained by distributed
algorithms.

A way to avoid the constraints of a centralized procedure is to
implement a decentralized one where each user calculates its estimation
of the optimal transmission power according to its local knowledge of
the system. This is, for example, the case in ad-hoc networks
applications. Most of the time, a distributed algorithm means an
iterative version of a centralized one. Mobiles update their power
allocation according to some rule based on the limited information
they retrieve from the system. Supposing that an optimal power
allocation exists, a distributed iterative algorithm is derived from a
differential equation in \cite{ieeeVT.foschini93} and its convergence
is proven analytically.  A distributed version of the algorithm of
\cite{ieeeVT.grandhi93} is presented in
\cite{ieeeCOM.grandhi94}. Building on these results, a general
framework for power control in cellular systems is given in
\cite{ieeeJSAC.yates95}. A review of different methods of centralized
and distributed power control in CDMA systems is given in
\cite{ieeeAP.elosery00}.

In this context, a natural framework is game theory, which studies
competition (as well as cooperation) between independent actors.
Tools of game theory have already been frequently used as a central
framework for modeling competition and cooperation in networking, see
for example \cite{paper.altman05} and references therein.  Building on
the framework of \cite{ieeeJSAC.yates95}, a game theoretic approach
was introduced in \cite{ieeePC.goodman00,ieeeCOM.mackenzie01}.
Numerous works on power allocation games have followed since, a selection
of which we present in Sec. \ref{sec:related}.

Game theory can be used to treat the case of any number of
players. However, as the size of the system increases, the number of
parameters increases drastically and it is difficult to gain insight
on the expressions obtained.

In order to obtain expressions depending only on few parameters, we
consider the system in an asymptotic setting, letting both the number
of users and the spreading factor tend to infinity with a fixed
ratio. We use tools of random matrix theory \cite{book.tulino04} to
analyze the system in this limit. Random matrix theory is a field of
mathematical physics that has been recently applied to wireless
communications to analyze various measures of interest such as
capacity or Signal to Interference plus Noise Ratio
(SINR). Interestingly, it enables to single out the main parameters of
interest that determine the performance in numerous models of
communication systems with more or less involved models of attenuation
\cite{ieeeIT.verdu99,ieeeIT.tse99,ieeeIT.shamai01,ieeeIT.tulino05}.
In addition, these asymptotic results provide good approximations
for the practical finite size case, as shown by simulations.

In the asymptotic regime, the non-cooperative game becomes a
non-atomic one, in which the impact (through interference) of any
single mobile on the performance of other mobiles is negligible.  In
the networking game context, the related solution concept is often
called Wardrop equilibrium \cite{paper.wardrop52}; it is often much
easier to compute than the original Nash equilibrium
\cite{paper.altman05}, and yet, the former equilibrium is a good
approximation for the latter, see details in
\cite{paper.haurie85}.  In this paper, we derive the non-atomic equilibrium, 
which generally corresponds to a non-uniform PA for the users. 

The non-atomic Nash equilibrium is studied in this paper for several
linear receivers, namely the matched filter and the MMSE filter, as
well as non-linear filters, such as the successive interference
cancellation (SIC) \cite{ieeeJSAC.muller01} version of those
filters. However, in order to perform SIC, the users need to know
their decoding order, in order to adjust their rates. In this paper,
we introduce ways of obtaining an ordering of the users in a
distributed manner. The ordering can be determined simply in a
distributed manner under weak hypotheses. This gives rise to a
different kind of power allocation, that depend explicitly on the
order in which the users are decoded.

Moreover, we quantify the gain of the non-uniform PA with respect
to uniform PA, according to the number of paths.  The originality of
the paper lies in the fact that we show that as the number of paths
increases, the optimal PA becomes more and more uniform due to the
ergodic behavior of all the CDMA channels. This is reminiscent of an
effect (``channel hardening'') already revealed in MIMO
\cite{ieeeIT.hochwald04}. The highest gain (in terms of utility) is
obtained in the case of flat fading (which also favors dis-uniform
power allocation between the users).

The layout of this paper is the following. First, a detailed account
of related works is made in Sec.~\ref{sec:related}. In order to be
self-contained, we introduce useful notations and concepts of random
matrix theory in Sec.~\ref{sec:notations}. The communication model
that will be used throughout the paper is detailed in
Sec.~\ref{sec:model}. Asymptotic SINR and capacity expressions are
given in Sec.~\ref{sec:SINR}. The particular game played between users
is introduced in Sec.~\ref{sec:nash}, along with the existence of a
Nash equilibrium. Finally, theoretical results for the power
allocation are derived in Sec.~\ref{sec:powerallocation} for unordered
users and Sec.~\ref{sec:SIC} when there is an ordering of the
users. Analytical results are matched with simulations in
Sec.~\ref{sec:simul}. Conclusions are provided in Sec.

\section{Related Work}
\label{sec:related}

This section is dedicated to present some of the works that use game
theory for power control. We remind that a Nash equilibrium is a
stable solution, where no player has an incentive to deviate
unilaterally, while a Pareto equilibrium is a cooperative dominating
solution, where there is no way to improve the performance of a player
without harming another one. Generally, both concepts do not coincide.
Following the general presentation of power allocation games in
\cite{ieeePC.goodman00,ieeeCOM.mackenzie01}, an abundance of works can
be found on the subject.

In particular, the utility generally considered in those articles is
justified in \cite{wcnc.rodriguez03} where the author describes a
widely applicable model ``from first principles''. Conditions under
which the utility will allow to obtain non-trivial Nash equilibria
(i.e., users actually transmit at the equilibrium) are derived. The
utility consisting of throughput-to-power ratio (detailed in
Sec. \ref{sec:nash}) is shown to satisfy these conditions. In
addition, it possesses a propriety of reliability in the sense that
the transmission occurs at non-negligible rates at the
equilibrium. This kind of utility function had been introduced in
previous works, with an economic leaning \cite{these.ji97,ieeeCOM.saraydar02}.

Unfortunately, Nash equilibria often lead to inefficient allocations,
in the sense that higher rates (Pareto equilibria) could  be  obtained for all mobiles if they
cooperated. To alleviate this problem, in addition to the non-cooperative game setting,
\cite{ieeeCOM.saraydar02} introduces a pricing strategy to force users to
transmit at a socially optimal rate. They obtain communication at Pareto equilibrium. 

In \cite{ieeeCOM.meshkati05}, defining the utility as advised in
\cite{wcnc.rodriguez03} as the ratio of the throughput to the
transmission power, the authors obtain results of existence and
unicity of a Nash equilibrium for a CDMA system. They extend 
this work to the case of multiple carriers in \cite{ieeeJSAC.meshkati06}.
In particular, it is shown that users will select and only transmit over their best carrier. 
As far as the attenuation is concerned, the consideration is restricted to flat fading  in
\cite{ieeeCOM.meshkati05} and in \cite{ieeeJSAC.meshkati06} 
(each carrier being flat fading in the latter).  However,
wireless transmissions generally suffer from the effect of multiple
paths, thus becoming frequency-selective. The goal of this
paper is to determine the influence of the number of paths (or the
selectivity of the channel) on the performance of PA.

This work is an extension of \cite{ieeeCOM.meshkati05} in
the case of frequency-selective fading, in the framework of multi-user
systems. We do not consider multiple carriers, as in
\cite{ieeeJSAC.meshkati06}, and the results are very different to
those obtained in that work. The extension is not trivial and involves
advanced results on random matrices with non-equal variances due to
Girko \cite{book.girko90} whereas classical results rely on the work
of Silverstein \cite{jma.silverstein95}. A part of this work was previously 
published as a conference paper \cite{rawnet.bonneau07}. 

Moreover, in addition to the linear filters studied in
\cite{ieeeCOM.meshkati05}, we study the enhancements provided by the
optimum and successive interference cancellation filters.

\section{Random Matrix Theory Notations and Concepts}
\label{sec:notations}

The following definitions and theorem can be found in
\cite{book.tulino04} and will be used in the following sections.
In this section, $N$ and $K$ are positive integers. 

\begin{definition} 
Let $\nu$ be a probability measure. The \emph{Stieltjes transform} $m^{\nu}$ associated to $\nu$ is given by
\begin{equation*}
m^{\nu}(z) = \int \frac{1}{t-z}\nu(dt).
\end{equation*}
\end{definition}

\begin{definition}
\label{def:ed}
Let $\vv=[\Vv_1,\dotsc,\Vv_N]$ be a vector. Its \emph{empirical
distribution} is the function $F_N^{\vv}:\RR\to[0,1]$ defined by:
\begin{equation*}
F_N^{\vv}(x)=\frac{1}{N}\#\{\Vv_i\leq x\, |\, i=1\dotsc N\}.
\end{equation*}
\end{definition}

In other words, $F_N^{\vv}(x)$ is the fraction of elements of $\vv$
that are inferior or equal to $x$.  In particular, if $\vv$ is the
vector of eigenvalues of a matrix $\VV$, $F_N^{\vv}$ is called the
\emph{empirical eigenvalue distribution} of $\VV$.

\begin{definition}
\label{def:ergodic}
Let $\VV$ be a $N\times K$ random matrix with independent columns and
entries $\Vv_{ij}$. Denote by $\loor{{\cdot}}$ the closest smaller
integer. $\VV$ is said to \emph{behave ergodically} if, as
$N,K\rightarrow\infty$ with $K/N\rightarrow \alpha$, for $x\in[0,1]$,
the empirical distribution of
\begin{equation*}
\left[\abs{\Vv_{\loor{xN},1}}^2,\dotsc,\abs{\Vv_{\loor{xN},K}}^2\right]
\end{equation*}
converges almost surely to a non-random limit distribution denoted $F_x^{\VV}({\cdot})$ and, for $y\in[0,\alpha]$, the empirical distribution of
\begin{equation*}
\left[\abs{\Vv_{1,\loor{yN}}}^2,\dotsc,\abs{\Vv_{N,\loor{yN}}}^2\right]
\end{equation*}
converges almost surely to a non-random limit distribution denoted $F_y^{\VV}({\cdot})$.
\end{definition}

\begin{definition}
\label{def:channelprofile}
Let $\VV$ be a $N\times K$ random matrix that behaves ergodically as
in Def. \ref{def:ergodic}, such as $F_x^{\VV}({\cdot})$ and
$F_y^{\VV}({\cdot})$ have all their moments bounded. The
\emph{two-dimensional channel profile} of $\VV$ is the function
$\rho^{\VV}(x,y):[0,1]\times[0,\alpha]\to\RR$ such that, if the random
variable $X$ is uniformly distributed in $[0,1]$, then the
distribution of $\rho^{\VV}(X,y)$ equals $F_y^{\VV}({\cdot})$ and, if
the random variable $Y$ is uniformly distributed in $[0,\alpha]$, then
the distribution of $\rho^{\VV}(x,Y)$ equals $F_x^{\VV}({\cdot})$.
\end{definition}

\begin{theorem}
\label{th:tulino}
Let $\YY=\VV\odot\WW$ be a $N\times K$ matrix, where ${\odot}$ is the
Hadamard (element-wise) product and $\VV$ and $\WW$ are independent
$N\times K$ random matrices. Assume that $\VV$ behaves ergodically
with channel profile $\rho^{\VV}(x,y)$ as in
Def. \ref{def:channelprofile} and that $\WW$ has i.i.d.~entries with
zero mean and variance $\frac{1}{N}$. Then, as $N,K\rightarrow\infty$
with $K/N\rightarrow\alpha$, the empirical eigenvalue distribution of
$\YY\YY^H$ converges almost surely to a non-random limit distribution
function whose Stieltjes transform is given by:
\begin{align*}
m^{\YY\YY^H}(z)&=\lim_{N\to\infty}\frac{1}{N}\Trace\left(\left(\YY\YY^H-z\Id\right)^{-1}\right)\\&=\int_0^1 u(x,z)dx
\end{align*}
and $u(x,z)$ satisfies the fixed point equation:
\begin{equation}
u(x,z)=\frac{1}{\int_0^{\alpha}\frac{\rho^{\VV}(x,y)dy}{1+\int_0^1 \rho^{\VV}(x',y)u(x',z)dx'}-z}.
\label{eq:mrztheorem}
\end{equation}
The solution to equation \eqref{eq:mrztheorem} exists and is unique in
the class of functions $u(x,z)\geq 0$, analytic for
$\textrm{Im}(z)>0$, and continuous on $x\in[0,1]$.
\end{theorem}

\section{Model}
\label{sec:model}

We consider a single uplink multi-user system cell, i.e., inter-cell
interference free case. The spreading length is denoted $N$. The
number of users in the cell is $K$.  The load is $\alpha=K/N$. The
general case of wide-band CDMA is considered where the signal
transmitted by user $k$ has complex envelope
\begin{equation*}
x_k(t)=\sum_n s_{kn}v_k(t-nT).
\end{equation*}
$v_k(t)$ is a weighted sum of elementary modulation pulses which
satisfy the Nyquist criterion with respect to the chip interval $T_c$
($T=NT_c$):
\begin{equation*}
v_k(t)=\sum_{\ell=1}^{N}v_{\ell k}\psi(t-(\ell-1) T_c).
\end{equation*}
The signal is transmitted over a frequency selective channel with
impulse response $c_k(\tau)$. Under the assumption of slowly-varying
fading, the continuous time received signal $y(t)$ at the base station
has the form:
\begin{equation*}
y(t)=\sum_{n}\sum_{k=1}^K s_{kn}\int c_k(\tau)v_k(t-nT-\tau)d\tau+n(t)
\end{equation*}
where $n(t)$ is zero-mean complex white Gaussian noise with variance
$\sigma^2$. The signal (after pulse matched filtering by $\psi^*(-t)$)
is sampled at the chip rate to get a discrete-time signal that has the
form:
\begin{equation}
\yy = \sum_{k=1}^K\HC_k\vv_k\powcoef{k}\signal_k + \NOISE
\label{eq:intro:frequencyselective}
\end{equation}
where $\HC_k$ are $N\times N$ Toeplitz matrices representing the
frequency selective fading for the $k$-th user, $\vv_k$ is a
$N\times1$ vector representing the spreading code of the $k$-th user,
and $\NOISE$ is an $N \times 1$ Additive White Gaussian Noise (AWGN)
vector with covariance matrix $\sigma^2\Id_N$.

We consider the case of a multipath channel. Under the assumption that
the number of paths from user $k$ to the base station is given by
$L_{k}$, the model of the channel is given by
\begin{equation}
c_{k}(\tau)=\sum_{\ell=0}^{L_{k}-1}\aaa{k}{\ell}\psi(\tau-\ttau{k}{\ell}).
\label{eq:channelmodel}
\end{equation}
where we assume that the channel is invariant during the time
considered. In order to compare channels at the same signal to
noise ratio, we constrain  the distribution of the i.i.d.~fading
coefficients $\aaa{k}{\ell}$ such as:
\begin{equation}
\Expect{}{\aaa{k}{\ell}}=0\ \text{and}\
\Expect{}{\abs{\aaa{k}{\ell}}^2}=\frac{\rrho}{L_{k}}.
\label{eq:firstsimplifyinghypothesis}
\end{equation}

Usually,  fading coefficients $\aaa{k}{\ell}$ are supposed to be
independent with decreasing variance as the delay increases.
In all cases, $\rrho$ is the average power of the channel, 
such as $\Expect{}{\abs{c_k(\tau)}^2}=\sum_{\ell=0}^{L_{k}-1}\Expect{}{\abs{\aaa{k}{\ell}}^2}=\rrho$,
for all channels considered.
For each user $k$, let  $\dd{k}{i}$ be the  Discrete
Fourier Transform of the fading process $c_{k}(\tau)$. The
frequency response of the channel at the receiver is given by:
\begin{equation}
\ddd_{k}(f)=\sum_{\ell=0}^{L_{k}-1}\aaa{k}{\ell} e^{-j2\pi f\ttau{k}{\ell}}\abs{\Psi(f)}^2. 
\label{eq:expressionofh}
\end{equation}
where we assume that the transmit filter $\Psi(f)$ and the receive
filter $\Psi^*(-f)$ are such that, given the bandwidth $W$,
\begin{equation}
\Psi(f)=\begin{cases} 1& \text{if}\ -\frac{W}{2}\leq f\leq\frac{W}{2}\\
0& \text{otherwise.} \end{cases}
\label{eq:expressionofpsi}
\end{equation}

Sampling at the various frequencies $f_1=-\frac{W}{2}$,
$f_2=-\frac{W}{2}+\frac{1}{N}W$, \dots,
$f_{N}=-\frac{W}{2}+\frac{N-1}{N}W$, we obtain the coefficients
$\dd{k}{i}$, $1\leq i\leq N$, as
\begin{equation}
\dd{k}{i}=\ddd_{k}(f_i)=\sum_{\ell=0}^{L_{k}-1}\aaa{k}{\ell}e^{-j2\pi\frac{i}{N}W\ttau{k}{\ell}}e^{j\pi W\ttau{k}{\ell}}.
\label{eq:dftchannelmodel}
\end{equation}
Note that $\Expect{}{\abs{\dd{k}{i}}^2}=\rrho$.

Since the users are supposed to be synchronized with the base station
and for sake of simplicity, we will consider in all the following that
users add a cyclic prefix of length equal to the channel impulse
response length to their code sequence.\footnote{Note that in the
asymptotic case (when $N \rightarrow \infty$), the result holds
without the need of a cyclic prefix as long as the channel is
absolutely summable \cite{book.Gray77}.} This case is similar to
uplink MC-CDMA \cite{pimrc.fazel93,isssta.lindner96}. As a
consequence, matrices $\{\HC_k\}$ are circulant
\cite{commag.bingham90} and can all be diagonalized in the Fourier basis
$\FF$ \cite{book.Gray77}. Model \eqref{eq:intro:frequencyselective} simplifies
therefore to:
\begin{equation}
\yy = \sum_{k=1}^K \FF\DD_k\FH\vv_k\powcoef{k}\signal_k + \NOISE
\label{eq:intro:diag}
\end{equation}
where $\DD_k$ is a diagonal matrix with diagonal elements
$\{\dd{k}{i}\}_{i=1\dotsc N}$.  For each user $k$, the coefficients
$\dd{k}{i}$ are the discrete Fourier transform of the channel impulse
response.  

We make the hypothesis that the users employ Gaussian i.i.d.~codes
with zero mean and variance $1/N$ \cite{ieeeIT.tse00}. This hypothesis
enables us to state simply our results, however almost all of the
results are valid for any distribution of the codes as long as it has
mean zero and variance $1/N$ \cite{ieeeIT.tulino05}. In particular,
since every unitary tranformation of a Gaussian i.i.d.~vector is a
Gaussian i.i.d.~vector (so that $\ww_i=\FH\vv_i$ has the same
distribution as $\vv_i$ for any $i$), we multiply $\yy$ in
\eqref{eq:intro:diag} with $\FH$ and obtain without any change in the
statistics:
\begin{align}
\yy &= \sum_{k=1}^K \DD_k\ww_k\powcoef{k}\signal_k + \NOISE \nonumber\\
&= \bigl(\HH\PP\odot\WW\bigr)\SIGNAL + \NOISE
\label{eq:intro:compact}
\end{align}
where $\odot$ is the Hadamard (element-wise) product.

In \eqref{eq:intro:compact}, $\HH$ is the frequency selective fading matrix, of size $N\times K$:
\begin{equation*}
\HH = \begin{bmatrix}\dd{1}{1}&\dd{2}{1}&\ldots&\dd{K}{1}\\
\vdots&\vdots&&\vdots\\
\dd{1}{N}&\dd{2}{N}&\ldots&\dd{K}{N}\end{bmatrix}.
\end{equation*}

$\PP$ is the root square of the diagonal power control  matrix, of size $K\times K$.

$\WW$ is an $N\times K$ random  spreading matrix:
\begin{equation*}
\WW = \bigl[\ww_1|\ww_2|\dotsb|\ww_K\bigr]\ \text{where}\ \ww_k = \begin{bmatrix}\uu{k}{1}\\ \vdots\\ \uu{k}{N}\end{bmatrix}.
\end{equation*}

Note that asymptotically (as $N \rightarrow \infty$), for a given
multipath channel of length $L$, model (\ref{eq:intro:compact}) is
also valid for the case of uplink DS-CDMA since all Toeplitz matrices
can be asymptotically diagonalized in a Fourier Basis
\cite{book.Gray77,ieeeIT.hachem02}.

In the following, we will assume that the frequency selective fading
matrix $\HH$ behaves ergodically, as in Def. \ref{def:ergodic}. The
two-dimensional channel profile of $\HH\PP$ is denoted
$\rho(f,x)=P(x)\abs{h(f,x)}^2,\ f\in[0,1],\ x\in[0,\alpha]$. $f$ is
the frequency index and $x$ is the user index. This enables us to use
Th. \ref{th:tulino} in order to obtain expressions for the SINR.

It is also assumed that the power of all users is upper bounded by $\Pmax$ 
and the square norm of the fading, on all paths, for all users, is upper bounded by $\hmax$.

\section{Asymptotic SINR Expressions}
\label{sec:SINR}

Let $\HHu{k}$ be the $k$-th column of $\HH$, and $\Hu{k}$ be $\HH$
with $\HHu{k}$ removed. Similarly, let $\WWu{k}$ be the $k$-th column
of $\WW$, and $\Wu{k}$ be $\WW$ with $\WWu{k}$ removed.  Let $\Pu{k}$
be $\PP$ with the $k$-th column and line removed. Finally, let
$\Gu{k}=\Hu{k}\Pu{k}\odot\Wu{k}$.

\subsection{Matched Filter}
\label{sec:Matchedfilter}

Supposing perfect CSI at the receiver, the matched filter for the
$k$-th user is given by
$\g_k=\powcoef{k}\left(\HHu{k}\odot\WWu{k}\right)$.  This leads to the
following expression for the SINR of user $k$
\begin{equation*}
\SINR_k = \frac{\abs{\g_k^H\g_k}^2}{\sigma^2\g_k^H\g_k+\g_k^H\left(\Gu{k}\GuH{k}\right)\g_k}.
\end{equation*}

\begin{proposition}
\label{prop:mf}
\cite{ieeeIT.tulino05}
As $N,K\rightarrow\infty$ with $K/N\rightarrow\alpha$, the SINR of user $k$ at the output of the matched filter is given by
\begin{equation*}
\SINR_k=\betamf\left(\frac{k}{N}\right)
\end{equation*}
where $\betamf: [0,\alpha]\to\RR$ is given by
\begin{multline}
\betamf(x)=P(x)\cdot\\{\frac{(H(x))^2}{\sigma^2 H(x)+\int_0^{\alpha}\int_0^1 P(y)\abs{h(f,y)}^2\abs{h(f,x)}^2dfdy}}
\label{eq:implicit:betamf}
\end{multline}
and $H(x)=\int_0^1\abs{h(f,x)}^2df$. 
\end{proposition}

Denoting $\SINR_k=\betamf_k$, Prop. \ref{prop:mf} enables us to
extract an approximation of the value of the SINR of user $k$ in the
finite size case
\begin{multline}
\betamf_k = \\ \frac{P_k\left(\frac{1}{N}\sum_{n=1}^N\abs{\dd{k}{n}}^2\right)^2}{\frac{\sigma^2}{N}\sum_{n=1}^N\abs{\dd{k}{n}}^2+\frac{1}{N^2}\sum_{j\neq k}\sum_{n=1}^N P_j\abs{\dd{j}{n}}^2\abs{\dd{k}{n}}^2}.
\label{eq:finite:betamf}
\end{multline}
We observe that
$P_k\frac{\partial\betamf_k}{\partial P_k}=\betamf_k$.

\subsection{MMSE Filter}
\label{sec:MMSEfilter}

Supposing perfect CSI at the receiver, the MMSE filter for the $k$-th
user is given by $\g_k^{\textrm{MMSE}}=\R^{-1}\g_k$, where
$\R=\left(\left(\HH\PP\odot\WW\right)\left(\HH\PP\odot\WW\right)^H+\sigma^2\Id_N\right)$.
This leads to the following expression for the SINR of user $k$ \cite{ieeeIT.tse99}
\begin{equation}
\SINR_k = \g_k^H\left(\Gu{k}\GuH{k}+\sigma^2\Id_N\right)^{-1}\g_k.
\label{eq:exactsinrmmse}
\end{equation}

\begin{proposition}
\label{prop:mmse}
\cite{ieeeIT.tulino05}
As $N,K\rightarrow\infty$ with $K/N\rightarrow\alpha$, the SINR of user $k$ at the output of the MMSE receiver is given by:
\begin{equation*}
\SINR_k=\betammse\left(\frac{k}{N}\right)
\end{equation*}
where $\betammse: [0,\alpha]\to\RR$ is a function defined by the implicit equation
\begin{equation}
\betammse(x)=P(x)\int_0^1\frac{\abs{h(f,x)}^2df}{\sigma^2+\int_0^{\alpha}\frac{P(y)\abs{h(f,y)}^2dy}{1+\betammse(y)}}.
\label{eq:implicit:beta}
\end{equation}
\end{proposition}

Denoting $\SINR_k=\betammse_k$, Prop. \ref{prop:mmse} enables us to
extract an approximation of the value of the SINR of user $k$ in the
finite size case
\begin{equation}
\betammse_k = P_k\frac{1}{N}\sum_{n=1}^N\abs{\dd{k}{n}}^2\frac{1}{\sigma^2+\frac{1}{N}\sum_{j\neq k} \frac{P_j\abs{\dd{j}{n}}^2}{1+\betammse_j}}.
\label{eq:finite:betammse}
\end{equation}

From \eqref{eq:exactsinrmmse}, we observe that $P_k\frac{\partial\betammse_k}{\partial P_k}=\betammse_k$.

From Prop. \ref{prop:mmse}, we have the capacity of user $k$
\begin{equation*}
\Gammmse_k = \frac{1}{N}\log_2\left(1+\betammse_k\right).
\end{equation*}
The global capacity of the system is
\begin{equation}
\Gammmse = \int_0^\alpha \log_2\left(1+\betammse(x)\right)dx.
\label{eq:gammmse}
\end{equation}

\subsection{Optimal Filter}
\label{sec:optimalfilter}

The term optimal filter designates a filter capable of decoding the
received signal at the bound given by Shannon's capacity. Hence it is
difficult to define an SINR associated to it. However, results of
random matrix theory can still be applied.  Let
$\YY=\left(\HH\PP\odot\WW\right)$. The definition of Shannon's
capacity per dimension for our system is
\begin{equation}
\Gamopt_{(N)}=\frac{1}{N}\log_2\det\left(\Id_N+\frac{1}{\sigma^2}\YY\YY^H\right).
\label{eq:shannoncapa}
\end{equation}
As $N,K\rightarrow\infty$ with $K/N\rightarrow\alpha$,
\begin{equation}
\Gamopt_{(N)}\rightarrow\int\log_2\left(1+\frac{1}{\sigma^2}t\right)\nu(dt)
\label{eq:sctendsto}
\end{equation}
where $\nu$ is the empirical eigenvalue distribution of $\YY\YY^H$, as
in Def. \ref{def:ed}.  If we differentiate the asymptotic value
$\Gamopt$ of \eqref{eq:sctendsto} with respect to $\sigma^2$, we
obtain
\begin{align}
\frac{\partial\Gamopt}{\partial\sigma^2}&=\log_2(e)\int\frac{-\frac{1}{\sigma^4}t}{1+\frac{1}{\sigma^2}t}\nu(dt)\nonumber\\
&=\log_2(e)\int\frac{\sigma^2\left(-\frac{1}{\sigma^4}t-\frac{1}{\sigma^2}+\frac{1}{\sigma^2}\right)}{\sigma^2\left(1+\frac{1}{\sigma^2}t\right)}\nu(dt)\nonumber\\
&=\log_2(e)\left(\int\frac{1}{t+\sigma^2}\nu(dt)-\frac{1}{\sigma^2}\int\nu(dt)\right)\nonumber\\
&=\log_2(e)\left(m^{\nu}(-\sigma^2)-\frac{1}{\sigma^2}\right)
\label{eq:scderived}
\end{align}
where $m^{\nu}(\cdot)$ is the Stieltjes transform of the empirical
eigenvalue distribution of $\YY\YY^H$. From Th. \ref{th:tulino},
$m^{\nu}(\cdot)$ is given by
\begin{equation*}
m^{\nu}(z)=\int_0^1 u(f,z)df
\end{equation*}
where $u(f,z)$ is given by \eqref{eq:mrztheorem} with
$\rho^{\HH\PP}(f,x)=\rho(f,x)=P(x)\abs{h(f,x)}^2$.  Given that if
$\sigma^2=+\infty$, $\Gamopt=0$, it is immediate to obtain $\Gamopt$
from \eqref{eq:scderived} as
\begin{equation}
\Gamopt = \log_2(e)\int_{\sigma^2}^{+\infty} m^{\nu}(-z)-\frac{1}{z}dz.
\label{eq:scint}
\end{equation}

\begin{proposition}
\label{prop:opt}
$\Gamopt$ and $\Gammmse$ are related through the following equality
\begin{align}
\Gamopt &= \Gammmse - \log_2(e)\int_0^\alpha\frac{\betammse(x)}{1+\betammse(x)}dx\nonumber\\
&\ +\int_0^1\log_2\left(1+\frac{1}{\sigma^2}\int_0^\alpha\frac{\rho(f,x)}{1+\betammse(x)}dx\right)df.
\label{eq:optmmseeq}
\end{align}
\end{proposition}

\begin{proof}
See Appendix \ref{sec:proof:opt}.
\end{proof}

The additional term in the right-hand side of \eqref{eq:optmmseeq}
corresponds to the non-linear processing gain. It quantifies the gain
in terms of capacity that can be achieved between pure linear MMSE and
non-linear filtering.

Assuming perfect cancellation of decoded users, successive
interference cancellation with MMSE filter achieves the optimum
capacity \cite{ieeeIT.muller01}.  The following proposition ensues
from this fact.

\begin{proposition}
\label{prop:sinropt}
\cite{ieeeIT.tulino05}
As $N,K\rightarrow\infty$ with $K/N\rightarrow\alpha$, the optimal capacity is given by:
\begin{equation*}
\Gamopt=\int_0^\alpha\log_2\left(1+\betasic(x)\right)dx
\end{equation*}
where $\betasic: [0,\alpha]\to\RR$ is a function defined by the implicit equation
\begin{equation}
\betasic(x)=P(x)\int_0^1\frac{\abs{h(f,x)}^2df}{\sigma^2+\int_0^{x}\frac{P(y)\abs{h(f,y)}^2dy}{1+\betasic(y)}}.
\label{eq:implicit:betasic}
\end{equation}
\end{proposition}

Prop. \ref{prop:sinropt} enables us to extract an expression that is
analog to the SINR for the optimal filter. Similarly to the case of
$\betammse$ in Sec. \ref{sec:MMSEfilter}, the derivative of
this expression obeys the property
$P_k\frac{\partial\betasic_k}{\partial P_k}=\betasic_k$.

\section{Games and Equilibria}
\label{sec:nash}

From now on, we denote $\SINR_k=\beta_k$, whichever filter is actually used.

\subsection{Power Allocation Game}

A game with a unique strategy set for all users is defined by a triple $\{S,\PA,(u_k)_{k\in S}\}$
where $S$ is the set of \emph{players},
$\PA$ is the set of \emph{strategies},
and $(u_k)_{k\in S}$ is the set of \emph{utility functions}, $u_k: \PA^{\abs{S}}\to \RR$. 

In our setting, the players are simply the users, indexed by the set
$S^K=\{1,\dotsc,K\}$. The strategy for a mobile is its power
allocation $P_k$, which we will assume belongs to a compact interval
$\PA=[0,\Pmax]\subseteq \RR$.  The utility measures the gain of a user
as a result of the strategy this user plays.  In
\cite{wcnc.rodriguez03}, the author derives what he calls Throughput
to Power Ratio (TPR) under minimal requirements. The utility of user
$k$ is expressed
\begin{equation}
u_k = \frac{\Gami_k}{P_k}.
\label{eq:genexpuserutility}
\end{equation}
We denote $\Gami_k=\Gami(\beta_k)$, where $\Gami(\cdot)$ is the same
function for all users. In \eqref{eq:genexpuserutility}, $\Gami$ is at
least $C^2$ and should satisfy conditions detailed in
\cite{wcnc.rodriguez03} in order to obtain an ``interesting''
equilibrium.

For example, in the simulations, we consider the goodput
$\Gami\left(\beta_k\right)$, which is proportional to
$\left(1-e^{-\beta_k}\right)^M$ where $M$ is the number of bits
transmitted in a CDMA packet. Remark that the usual definition of
goodput would rather be considered proportional to
$q(\beta_k)=(1-\textrm{BER}_k)^M$, where $\textrm{BER}$ is the bit error rate.
However, this quantity is not zero when the transmitted power is
zero. Using this function in the utility would lead to the
unsatisfying conclusion that mobiles should not transmit at all, since
the (improbable) event of a correct guess gives them infinite utility
\cite{ieeePC.goodman00}.  Therefore, an adapted version of the goodput
is adopted, where a factor 2 is added before the BER. The performance
measure considered is hence proportional to
$q_2(\beta_k)=(1-2\textrm{BER}_k)^M$, leading to the expression
above. This function has the desirable property $q_2(0)=0$ and its
shape follows closely the shape of the original goodput
$q(\cdot)$. This is a relevant performance measure, as each mobile
wants to use its (limited) battery power to transmit the maximum
possible amount of information.

This utility is expressed in \emph{bits per joule}. In the
non-cooperative game setting, each user wants to selfishly maximize
its utility. A Nash equilibrium is obtained when no user can benefit
by unilaterally deviating from its strategy.

To obtain the maximum utility achievable by user $k$, we differentiate
$u_k$ with respect to the power $P_k$ and equate to 0. We obtain
\begin{equation}
P_k\frac{\partial\beta_k}{\partial P_k}\Gami'(\beta_k)-\Gami(\beta_k)=0.
\label{eq:deriveequal0}
\end{equation}
For all filters under consideration, \eqref{eq:implicit:betamf}, \eqref{eq:implicit:beta} and \eqref{eq:implicit:betasic} 
imply $P_k\frac{\partial\beta_k}{\partial P_k}=\beta_k$,
thus \eqref{eq:deriveequal0} reduces an equation on $\beta_k$
\begin{equation}
\beta_k\Gami'(\beta_k)-\Gami(\beta_k)=0.
\label{eq:equationbetastar}
\end{equation}

Eq. \eqref{eq:equationbetastar} is particularly interesting in the
case when there exists a unique solution $\betastar$.

The existence of a solution to \eqref{eq:equationbetastar} is
guaranteed as long as the function $\Gami(\cdot)$ is a quasiconcave
function of the SINR, i.e., there exists a point below which the
function is non-decreasing, and above which the function is
non-increasing \cite{ieeeCOM.saraydar02,wcnc.rodriguez03}.  In
addition, we assume that the function $\Gami(\cdot)$ takes value
$\Gami(0)=0$, so that users cannot achieve an infinite utility by not
transmitting.  This occurs for several functions $\Gami(\cdot)$ of
interest, in particular the goodput \cite{ieeeCOM.meshkati05}, which
we will use for simulations. Unfortunately, the capacity can not be
used as a function $\Gami(\cdot)$, since it leads to the trivial
result $\betastar = 0$ for this utility function. The uniqueness of
the solution $\betastar$ to \eqref{eq:equationbetastar} is due to fact
that the SINR of each user is a strictly increasing function of its
transmit power.  Given the target SINR $\betastar$, we obtain the
strategy of users in the next section.

\section{Power Allocation in the Nash Equilibrium}
\label{sec:powerallocation}

\subsection{Flat Fading}

In this subsection, we show that the results of
\cite{ieeeCOM.meshkati05} for Matched and MMSE filters are a special case of our setting when
$L=1$ (flat fading case). In addition, we derive the power allocation
for the Optimum filter.  When there is only one path, for each user
$k$, denoted by its index $\frac{k}{N}=x\in[0,\alpha]$, $h(f,x)$ does
not depend on $f$. Given the target SINR $\betastar$, we have explicit
expressions of the power with which user $k$ transmits for the various
receivers.  

In Appendix \ref{sec:proof:large}, we show that the influence of the strategy of a
player on the payoffs of other players is (asymptotically) ``small''.
It justifies the fact that we can obtain an equilibrium in the
asymptotic setting, without the need for players to possess all the
information on the system. Their local information is sufficient. In the asymptotic limit, we obtain results
similar to Wardrop equilibrium: the strategy used by each user does
not influence the strategy of other users.

\subsubsection{Matched filter} From Prop. \ref{prop:mf}, the continuous formulation is
\begin{equation*}
P(x) = \frac{\betastar\left(\sigma^2+\int_0^\alpha P(y)\abs{h(y)}^2dy\right)}{\abs{h(x)}^2}
\end{equation*}
or equivalently in a discrete form
\begin{equation}
P_k = \frac{\betastar\left(\sigma^2+\frac{1}{N}\sum_{j=1,j\neq k}^K P_j \abs{h_j}^2\right)}{\abs{h_k}^2}.
\label{eq:mfflatdiscrete}
\end{equation}
Summing \eqref{eq:mfflatdiscrete} over $k=1,\dotsc,K$, we obtain a closed form expression for the minimum power with which user $k$ transmits when using the matched filter
\begin{equation}
P_k = \frac{1}{\abs{h_k}^2}\frac{\sigma^2\betastar}{1-\alpha{\betastar}}\ \textrm{for}\ \alpha<\frac{1}{\betastar}.
\label{eq:mfflatminpower}
\end{equation}

\subsubsection{MMSE filter} From Prop. \ref{prop:mmse}, the continuous formulation is
\begin{equation*}
P(x) = \frac{\betastar\left(\sigma^2+\frac{1}{1+\betastar}\int_0^\alpha P(y)\abs{h(y)}^2 dy\right)}{\abs{h(x)}^2}.
\end{equation*}
or equivalently in a discrete form
\begin{equation}
P_k = \frac{\betastar\left(\sigma^2+\frac{1}{1+\betastar}\frac{1}{N}\sum_{j=1,j\neq k}^K P_j \abs{h_j}^2\right)}{\abs{h_k}^2}.
\label{eq:flatdiscrete}
\end{equation}
Summing \eqref{eq:flatdiscrete} over $k=1,\dotsc,K$, we obtain a closed form expression for the minimum power with which user $k$ transmits when using the MMSE filter
\begin{equation}
P_k = \frac{1}{\abs{h_k}^2}\frac{\sigma^2\betastar}{1-\alpha\frac{\betastar}{1+\betastar}}\ \textrm{for}\ \alpha<1+\frac{1}{\betastar}.
\label{eq:flatminpower}
\end{equation}
Both \eqref{eq:mfflatminpower} and \eqref{eq:flatminpower} are the same results as in \cite{ieeeCOM.meshkati05}. 

\subsubsection{Optimum filter} Each user maximizes its utility for a 
SINR equal to $\betastar$. However, in the case of the optimum filter,
the SINR is not defined directly.  It is nevertheless possible to
extract an equivalent quantity from the expression of the capacity,
since the value of the capacity of user $k$ at the equilibrium is
given by $\Gamoptstar = \frac{1}{N}\log_2\left(1+\betastar\right)$.
\begin{proposition}
\label{prop:optflat}
The power allocation is given by 
\begin{equation}
P_k = \frac{1}{\abs{h_k}^2}\frac{\sigma^2\betaplus}{1-\alpha\frac{\betaplus}{1+\betaplus}}\ \textrm{for}\ \alpha<1+\frac{1}{\betaplus}
\label{eq:optflatminpower}
\end{equation}
where $\betaplus$ is the solution to 
\begin{multline}
\alpha\log_2\left(1+\betaplus\right)-\alpha\log_2(e)\frac{\betaplus}{1+\betaplus}\\
+\log_2\left(1+\frac{1}{1+\betaplus}\frac{\alpha\betaplus}{1-\alpha\frac{\betaplus}{1+\betaplus}}\right) = \alpha\log_2\left(1+\betastar\right). 
\label{eq:optmmsebetaplusprop}
\end{multline}
\end{proposition}

\begin{proof}
See Appendix \ref{sec:proof:optflat}.
\end{proof}

\subsection{Frequency Selective Fading}

In the context of frequency selective fading, for each user $k$, denoted by its index $\frac{k}{N}=x\in[0,\alpha]$, there are $L>1$ paths with respective attenuations 
$h_{\ell}(x),\ \ell=1,\dotsc,L$, which are i.i.d.~random variables with some known distribution.
We suppose that $h_{\ell}(x)$ has mean zero, and the distributions of the real part and imaginary part of $h_{\ell}(x)$ are even functions, as for example the Gaussian distribution,
which we consider in the simulations. 
$h(f,x)$ depends on $f$ through
$h(f,x)=\sum_{\ell=1}^{L}h_{\ell}(x)e^{-2\pi i f (\ell-1)}$. 
Given the target SINR $\betastar$, the Nash equilibrium power allocation is determined by implicit equations for the various receivers.
 
\subsubsection{Matched filter} The continuous formulation is
\begin{multline*}
P(x) = \betastar\cdot\\ \frac{\sigma^2 H(x)+\int_0^1\int_0^\alpha P(y)\abs{h(f,y)}^2\abs{h(f,x)}^2dfdy}{\left(H(x)\right)^2}
\end{multline*}
or equivalently in a discrete form
\begin{multline}
P_k = \betastar\cdot\\ \frac{\frac{\sigma^2}{N}\sum_{n=1}^N\abs{\dd{k}{n}}^2+\frac{1}{N}\sum_{n=1}^N\abs{\dd{k}{n}}^2\frac{1}{N}\sum_{j\neq k}^K P_j\abs{\dd{j}{n}}^2}{\left(\frac{1}{N}\sum_{n=1}^N\abs{\dd{k}{n}}^2\right)^2}.
\label{eq:mffsdiscrete}
\end{multline}
In \eqref{eq:mffsdiscrete}, $\dd{k}{n}=h\left(\frac{n-1}{N},\frac{k}{N}\right)$.

In this expression, the power allocation of user $k$ seems to depend
on the power allocation and fading realization of all the other
users. However, when the number of users tends to infinity, the
strategy of any single user does not have any influence on the payoff
of user $k$, as shown in Appendix \ref{sec:proof:large}. Hence, the
appropriate framework is non-atomic games.  The expression
$\frac{1}{N}\sum_{j=1}^K P_j\abs{\dd{j}{n}}^2$ is asymptotically a
constant (not depending on $n$), denoted $\Cste$.
\begin{equation}
\Cste = \frac{\alpha\betastar\sigma^2\frac{1}{K}\sum_{j=1}^K\frac{\abs{\dd{j}{n}}^2}{\E_j}}{1-\alpha\betastar\frac{1}{K}\sum_{j=1}^K\frac{\abs{\dd{j}{n}}^2}{\E_j}}
\label{eq:cstemf}
\end{equation}
where $\E_j = \frac{1}{N}\sum_{m=1}^N\abs{\dd{j}{m}}^2$. 

As $K\to\infty$, we can apply the Central Limit Theorem to the sum of random variables
\begin{equation}
\frac{1}{K}\sum_{j=1}^K\frac{\abs{\dd{j}{n}}^2}{\E_j}.
\label{eq:converge}
\end{equation}
It tends to its expectation, which is equal to $1$ (see Appendix \ref{sec:proof:converge}). 

It follows that asymptotically $\Cste=\frac{\alpha\betastar\sigma^2}{1-\alpha\betastar}$
(and simulations in Sec. \ref{sec:simul} prove that this approximation is valid for moderate finite values of $N$). 
From \eqref{eq:mffsdiscrete}, we obtain a formula similar to \eqref{eq:mfflatminpower}
\begin{equation}
P_k = \frac{1}{\E_k}\frac{\sigma^2\betastar}{1-\alpha{\betastar}}\ \textrm{for}\ \alpha<\frac{1}{\betastar}.
\label{eq:mffsminpower}
\end{equation}

\subsubsection{MMSE filter} The continuous formulation is
\begin{equation}
P(x) = \frac{\betastar}{\int_0^1\frac{\abs{h(f,x)}^2 df}{\sigma^2+\frac{1}{1+\betastar}\int_0^\alpha P(y)\abs{h(f,y)}^2 dy}}
\label{eq:fscontinuous}
\end{equation}
or equivalently in a discrete form
\begin{equation}
P_k = \frac{\betastar}{\frac{1}{N}\sum_{n=1}^N\frac{\abs{\dd{k}{n}}^2}{\sigma^2+\frac{1}{1+\betastar}\frac{1}{N}\sum_{j=1,j\neq k}^K P_j \abs{\dd{j}{n}}^2}}.
\label{eq:fsdiscrete}
\end{equation}
In \eqref{eq:fsdiscrete}, $\dd{k}{n}=h\left(\frac{n-1}{N},\frac{k}{N}\right)$.

As previously, when the number of users tends to infinity,
$\frac{1}{N}\sum_{j=1}^K P_j\abs{\dd{j}{n}}^2$ is asymptotically a
constant (not depending on $n$), denoted $\Cste$.
\begin{equation}
\Cste = \frac{\alpha\betastar\sigma^2\frac{1}{K}\sum_{j=1}^K\frac{\abs{\dd{j}{n}}^2}{\E_j}}{1-\frac{\alpha\betastar}{1+\betastar}\frac{1}{K}\sum_{j=1}^K\frac{\abs{\dd{j}{n}}^2}{\E_j}}
\label{eq:cstemmse}
\end{equation}
where $\E_j = \frac{1}{N}\sum_{m=1}^N\abs{\dd{j}{m}}^2$.

It follows that asymptotically $\Cste=\frac{\alpha\betastar\sigma^2}{1-\alpha\frac{\betastar}{1+\betastar}}$, we obtain a formula similar to \eqref{eq:flatminpower}
\begin{equation}
P_k = \frac{1}{\E_k}\frac{\sigma^2\betastar}{1-\alpha\frac{\betastar}{1+\betastar}}\ \textrm{for}\ \alpha<1+\frac{1}{\betastar}.
\label{eq:fsminpower}
\end{equation}

\subsubsection{Optimum filter} Each user maximizes its utility for a 
SINR equal to $\betastar$. However, in the case of the optimum filter,
the SINR is not defined directly.  It is nevertheless possible to
extract an equivalent quantity from the expression of the capacity,
since the value of the capacity of user $k$ at the equilibrium is
given by $\Gamoptstar = \frac{1}{N}\log_2\left(1+\betastar\right)$.
\begin{proposition}
\label{prop:optfs}
Asymptotically, as $N,K\rightarrow\infty$, the power allocation is given by 
\begin{equation}
P_k = \frac{1}{\E_k}\frac{\sigma^2\betaplus}{1-\alpha\frac{\betaplus}{1+\betaplus}}\ \textrm{for}\ \alpha<1+\frac{1}{\betaplus}
\label{eq:optfsminpower}
\end{equation}
where $\betaplus$ is the solution to 
\begin{multline}
\alpha\log_2\left(1+\betaplus\right)-\alpha\log_2(e)\frac{\betaplus}{1+\betaplus}\\
+\log_2\left(1+\frac{1}{1+\betaplus}\frac{\alpha\betaplus}{1-\alpha\frac{\betaplus}{1+\betaplus}}\right) = \alpha\log_2\left(1+\betastar\right). 
\label{eq:optmmsebetapluspropfs}
\end{multline}
\end{proposition}

\begin{proof}
The proof is similar to the proof of Prop. \ref{prop:optflat}. 
\end{proof}

We observe that for all filters considered, the optimal PA is a
constant times the inverse of the \emph{total} energy of the channel
$\E_j$. Via Parseval's Theorem, $\E_j = \sum_{\ell=1}^L
\abs{h_{\ell}\left(\frac{j}{N}\right)}^2$. It is a sum of
i.i.d. random variables. As the number of paths increases, the optimal
PA tends to a uniform PA.  This is an effect similar to ``channel
hardening'' \cite{ieeeIT.hochwald04}: as the number of paths
increases, the variance of the distribution of the channel energy
decreases and the Nash equilibrium PA becomes more and more uniform
for all users.

\section{Successive Interference Cancellation}
\label{sec:SIC}

The optimal filter gives a bound on the performance that can be
achieved through (non-linear) filtering at the base station. In order
to improve the performance of the system, we introduce Successive
Interference Cancellation (SIC) \cite{ieeeJSAC.muller01} at the base
station. Under the assumption of perfect decoding, SIC improves
immensely the performance of linear filters (Matched Filter or MMSE
Filter). The MMSE SIC filter actually achieves the optimum filter
bound, under the assumption of perfect decoding. The principle of SIC
receivers is quite simple: users are ordered and are decoded
successively. At each step, supposing that the user has been encoded
at the appropriate decoding rate, the signal is decoded and its
contribution to the interference is then perfectly subtracted. This
removes some of the inter-user interference and therefore increases
the $\SINR$ of the following decoded users.

The challenge is that the users must transmit at the appropriate rate
to avoid the catastrophic occurrence of imperfect decoding.  Usually,
the ordering of users is done in a centralized way, at the base
station which then advertises it to the users. However, for the
protocol to remain distributed, users should be able to decide,
based on their local information, at which rate to transmit. 

At equilibrium, the rate is determined by the SINR $\betastar$, and it
is the transmission power of the user that is determined according to
its rank of decoding. The equilibrium PA can be determined in a simple
manner when the number of multipaths is finite ($L<\infty$) and the
number of users is very high ($K\to\infty$).  In
Sec.~\ref{sec:wholelaw}, we make use of the fact that the whole law of
$\E_j$ is realized in this case, so that users automatically know
their rank of decoding.  Another manner to give a (random) ordering of
decoding is to introduce an additional degree of liberty in the
system. In Sec.~\ref{sec:correlated}, we develop a correlated game
framework that enables users to learn their rank of decoding in a
simple way. In the following, we assume that each user has a unique
has a unique i.d.~number $j$ ranging between 1 to $K$. 

\subsection{Ordering when $K\to\infty$}
\label{sec:wholelaw}

 
If the number of users $K\to\infty$, with $L$ fixed, the whole law of
the total channel energy will be realized.  Assume the base station
advertises to the users that they will be decoded by decreasing total
channel energy.  Each user knows, according to the realization of its
fading, its rank in the decoding order given by $K$ times $1$ minus
the cumulative distribution function $D(\cdot)$ of the total channel
energy $\E_j$.
\begin{equation*}
\textrm{rank}_j = K (1-D(\E_j)).
\end{equation*}

In case that the base station
advertises to the users that they will be decoded by increasing total
channel energy, user $j$  will have rank $\textrm{rank}_j = K D(\E_j)$.


\subsection{Correlated Equilibrium}
\label{sec:correlated}

We wish to introduce a simple mechanism that enables players
to coordinate and to know in which order they will be decoded. 
We place ourselves in the context of correlated
games. The notion of correlated equilibrium was introduced by
R. Aumann\footnote{Prof. R. Aumann has received in 2005 the Nobel
prize in economy for his contributions to game theory, together with
Thomas Schelling.} in \cite{paper.aumann74} and further studied in
\cite{paper.aumann87,paper.hart89,paper.neyman97}.  They represent a 
generalization of Nash equilibrium. The important
feature of correlated games is the presence of an
\emph{arbitrator}. An arbitrator needs not have any intelligence or
knowledge of the game, it needs only to send random (private or
public) signals to the players that are independent of all other data
in the game. In the context of non-cooperative games, each player has
the possibility not to consider the signal(s) it
receives. Coordination between players turns out to be useful also in
the case of cooperative optimization.  The signals enable joint
randomization between the strategies of the players, possibly
resulting in equilibria with higher payoffs.  The concept of
correlated games was recently introduced in a networking context in
\cite{networking.altman06}, where the authors consider a simple ALOHA
setting.

The simplest and most intuitive coordination mechanism is given by a
common signal which users as well as the base station overhear before
each transmission. There are $K!$ possible permutations of $K$
users. Hence, the arbitrator broadcasts a signal to the users
belonging to the set $\{0,\dotsc,K!-1\}$. Each of these numbers
corresponds to a permutation $\pi$ of $\{1,\dotsc,K\}$ that gives the
(random) ordering of decoding as $\textrm{rank}_j=\pi(j)$. The users
can then adjust their transmit power according to this ordering.  In
terms of size of the message, this is equivalent to the case when the
base station decides the decoding order and broadcasts it to the
users, or sends $K$ individual messages of $\ln(K)$ bits containing
the rank, since $\ln(K!)=K\ln(K)+o(K\ln(K))$. However, there is no
need of either any knowledge of the system or computations at the base
station in the case of the correlated mechanism.

\subsection{SIC Power Allocations}

In both cases, once the users know their order, they can calculate
their transmit power according to the filter that is used. The
equilibrium still occurs when all users reach the SINR $\betastar$. A
single user will not benefit by deviating, since it would decrease its
utility. From now on, index $k$ denotes the rank of decoding.

In the case of the matched filter with SIC, the SINR of the user decoded at rank $k$ is
\begin{multline}
\betamf_k =\\  \frac{P_k\left(\frac{1}{N}\sum_{n=1}^N\abs{\dd{k}{n}}^2\right)^2}{\frac{\sigma^2}{N}\sum_{n=1}^N\abs{\dd{k}{n}}^2+\frac{1}{N^2}\sum_{j>k}\sum_{n=1}^N P_j\abs{\dd{j}{n}}^2\abs{\dd{k}{n}}^2}.
\label{eq:finite:betamfsic}
\end{multline}
From \eqref{eq:finite:betamfsic}, we get the equilibrium PA of user $k$ as
\begin{multline}
P_k = \betastar\cdot\\\frac{\frac{\sigma^2}{N}\sum_{n=1}^N\abs{\dd{k}{n}}^2+\frac{1}{N^2}\sum_{j>k}\sum_{n=1}^N P_j\abs{\dd{j}{n}}^2\abs{\dd{k}{n}}^2}{\left(\frac{1}{N}\sum_{n=1}^N\abs{\dd{k}{n}}^2\right)^2}.
\label{eq:finite:pamfsic}
\end{multline}
In the case of the MMSE filter with SIC, the SINR of the user decoded at rank $k$ is 
\begin{equation}
\betammse_k = P_k\frac{1}{N}\sum_{n=1}^N\abs{\dd{k}{n}}^2\frac{1}{\sigma^2+\frac{1}{N}\sum_{j>k} \frac{P_j\abs{\dd{j}{n}}^2}{1+\betammse_j}}.
\label{eq:finite:betammsesic}
\end{equation}
From \eqref{eq:finite:betammsesic}, we get the equilibrium PA of user $k$ as
\begin{equation}
P_k = \frac{\betastar}{\frac{1}{N}\sum_{n=1}^N\frac{\abs{\dd{k}{n}}^2}{\sigma^2+\frac{1}{1+\betastar}\frac{1}{N}\sum_{j>k}^K P_j \abs{\dd{j}{n}}^2}}.
\label{eq:finite:pammsesic}
\end{equation}

For flat fading, a simple recursion gives the equilibrium PA (see Appendix \ref{sec:PAsicflat}).
We obtain respectively
\begin{gather}
P_{k}^{\textrm{MF}}=\frac{\sigma^2\betastar}{\abs{h_{k}}^2}\left(1+\frac{1}{N}\betastar\right)^{K-k},\label{eq:pamfsicflat}\\
P_{k}^{\textrm{MMSE}}=\frac{\sigma^2\betastar}{\abs{h_{k}}^2}\left(1+\frac{1}{N}\frac{\betastar}{1+\betastar}\right)^{K-k}.\label{eq:pammsesicflat}
\end{gather}

As far as frequency-selective fading is concerned, this gives us the
form of the asymptotic expressions. Asymptotically, the power
allocation of one user will not depend on the PA of the other users,
as shown in Appendix \ref{sec:proof:large}. With a similar reasoning as in
Sec.~\ref{sec:powerallocation}, the expressions mimic
\eqref{eq:pamfsicflat} and \eqref{eq:pammsesicflat} with the total channel energy $\E_k$
replacing $\abs{h_k}^2$, i.e.,
\begin{gather}
P_{k}^{\textrm{MF}}=\frac{\sigma^2\betastar}{\E_{k}}\left(1+\frac{1}{N}\betastar\right)^{K-k},\label{eq:pamfsic}\\
P_{k}^{\textrm{MMSE}}=\frac{\sigma^2\betastar}{\E_{k}}\left(1+\frac{1}{N}\frac{\betastar}{1+\betastar}\right)^{K-k}.\label{eq:pammsesic}
\end{gather}
These expressions are also validated by simulations. 

Since MMSE SIC with perfect decoding is equivalent to the optimum
filter, we thus obtain a second possible equilibrium PA for the
optimum filter. In Sec.~\ref{sec:simul}, we investigate which is the
PA which minimizes total amount of power needed to transmit at
equilibrium SINR. In the case of automatic ordering of the users, one
question is whether it is best to order the users by increasing or
decreasing total fading energy. The answer is the following: it is
always best to decode the users by decreasing total channel energy
$\E_1<\dotsb<\E_k$ (see Appendix \ref{sec:ordering}).

An interesting feature of equilibrium PA \eqref{eq:pamfsic} and
\eqref{eq:pammsesic} is that there is no limitation on the number of
users than can be accomodated by the system, contrary to the previous
case of \eqref{eq:mffsminpower}, \eqref{eq:fsminpower} and
\eqref{eq:optfsminpower}. The limitation is only imposed by the increasing
power needed for each new user decoded last, which grows without bound
as an exponential.


\section{Numerical Results}
\label{sec:simul}

In all the following, we consider that $\Pmax$ is chosen sufficiently
high so that users can actually transmit at the equilibrium PA
values. For the simulations, we consider the usual case of Rayleigh
fading.  Although Rayleigh distribution is not bounded from above,
simulations show that the results still hold.

We consider a CDMA system with $K=32$ users and a spreading factor
$N=256$. The noise variance is $\sigma^2=10^{-10}$. For a number of
bits in a CDMA packet $M=100$, the goodput is
$\Gami(\beta)=\left(1-e^{-\beta}\right)^{100}$ (see
\cite{ieeeCOM.meshkati05}), and $\betastar=6.48$. The capacity
achieved at the Nash Equilibrium is
$C=\alpha\log_2\left(1+\betastar\right)=0.39$ bits/s. Unfortunately,
the capacity itself cannot be used as a relevant performance measure
in the definition of the utility, because in this case the maximal
utility is obtained when not sending.


We have performed simulations over 10000
realizations. Fig.~\ref{fig:theoretic} shows the good fit of theoretic
values calculated directly from \eqref{eq:mffsminpower},
\eqref{eq:fsminpower} and \eqref{eq:optfsminpower} with those
simulations. The values of the utility do not depend on the number of
multipaths. We see that optimum filter requires the minimal power, and
matched filter the maximal power to achieve the required goodput.

In Fig.~\ref{fig:1} we have plotted the average utility versus the
number of multipaths $L$. Multipaths are supposed to be
i.i.d.~Rayleigh distributed with variance $1/L$, in order for the
channels to have the same energy. Two cases are considered: the
utility obtained in the Nash equilibrium, according to the PA given by
\eqref{eq:mffsdiscrete} and \eqref{eq:fsdiscrete}, and the utility in
the case where all nodes transmit at the same power. For comparison
purposes, the sum of the uniform powers is equal to the sum of the
powers used in the Nash equilibrium. In addition, simulations (not
reproduced here) show that this value gives the higher average utility
for a uniform PA.
The utility does not vary with
$L$ in the Nash equilibrium: the Central Limit Theorem applies to the
utility, which is a constant times the random variable $\E_k$ in the
Nash equilibrium. The utility with uniform powers is always inferior
to the utility in the Nash equilibrium. However, as $L$ increases, the
gap decreases, as the variance of $\E_k$ decreases, and the
equilibrium PA becomes uniform.

In Fig.~\ref{fig:powerallocs} we have plotted the average of the
inverse power of the users in the Nash equilibrium for each of the
investigated schemes. We plot the average inverse power because of the
direct relation to the utility for the users. The higher this average,
the higher the utility for the user. The SIC filters are always more
efficient than their linear counterparts. However, for a load
$\alpha<0.12$ and optimum filter\footnote{The value of $\alpha$ is
obtained as solution of the equation
$\alpha\betastar\frac{\betastar}{1+\betastar}(1-\alpha\frac{\betaplus}{1+\betaplus})=\betaplus(1-\exp(-\alpha\frac{\betastar}{1+\betastar}))$.},
it is better to use the first variation of PA \eqref{eq:optfsminpower}
than use MMSE SIC \eqref{eq:pammsesic}. This relation is reversed when
$\alpha>0.12$. In addition to the theoretical curves, Monte-Carlo
simulations were performed both with random ordering (circles) and
ordering by decreasing total channel energy (crosses), for $L=8$
multipaths. Simulations show that the optimal ordering improves the
power efficiency of the successive interference cancellation filters.

In Fig.~\ref{fig:ordering}, we investigate the amelioration provided
by optimal ordering as a function of the number of multipaths. The
simulations are done for $K=128$ users, in order to be in the
``interesting'' zone $\alpha>0.12$.  As expected, as the number of
paths increases, the total channel energy is more and more the same
for each channel and the gain provided by ordering the users
decreases. However, when the number of users is very large and they
benefit from automatic ordering, we see that the utility with the MMSE
SIC equilibrium PA is the maximal utility that can be obtained in the
non-cooperative setting.

\begin{figure}
\centerline{\includegraphics[width=9cm]{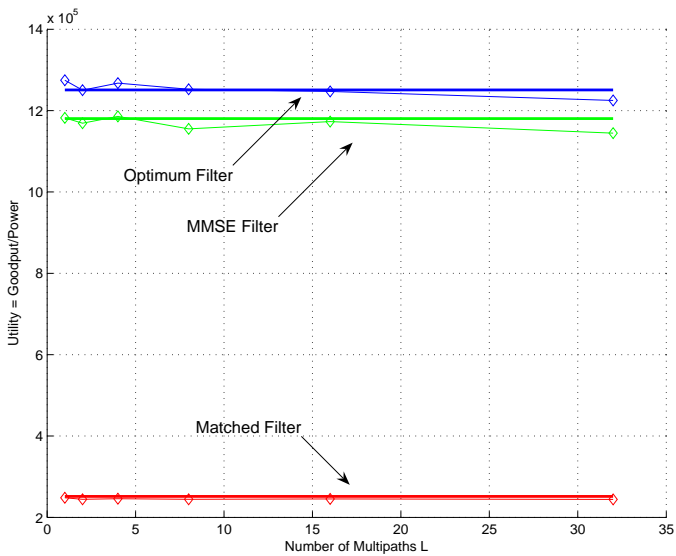}}
\caption{Comparison of theoretic values and simulations for utilities in the Nash equilibrium.}
\label{fig:theoretic}
\end{figure}

\begin{figure}
\centerline{\includegraphics[width=9cm]{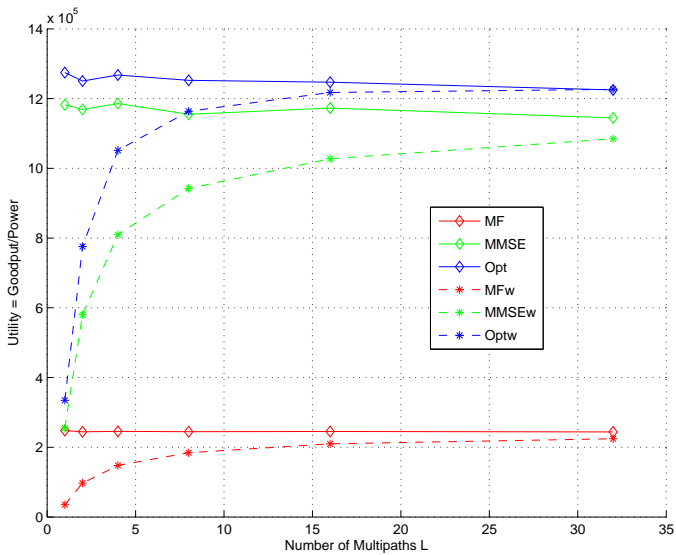}}
\caption{Simulation of utilities in the Nash equilibrium and constant power allocations versus $L$.}
\label{fig:1}
\end{figure}

\begin{figure}
\centerline{\includegraphics[width=9cm]{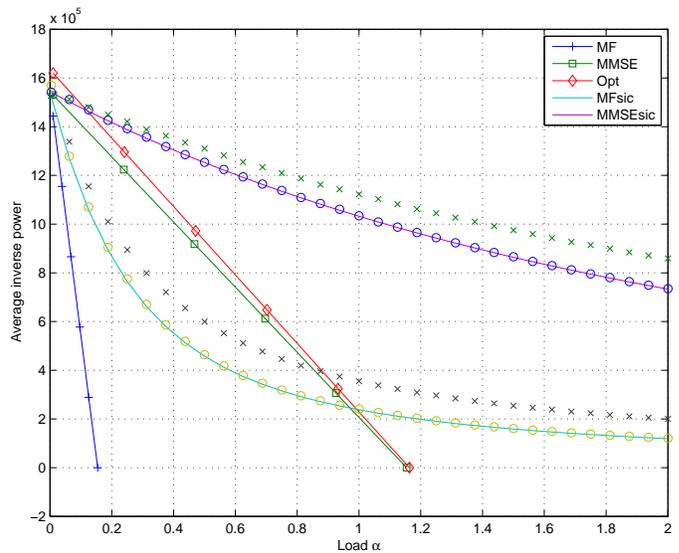}}
\caption{Average inverse power used by the different filters.}
\label{fig:powerallocs}
\end{figure}

\begin{figure}
\centerline{\includegraphics[width=9cm]{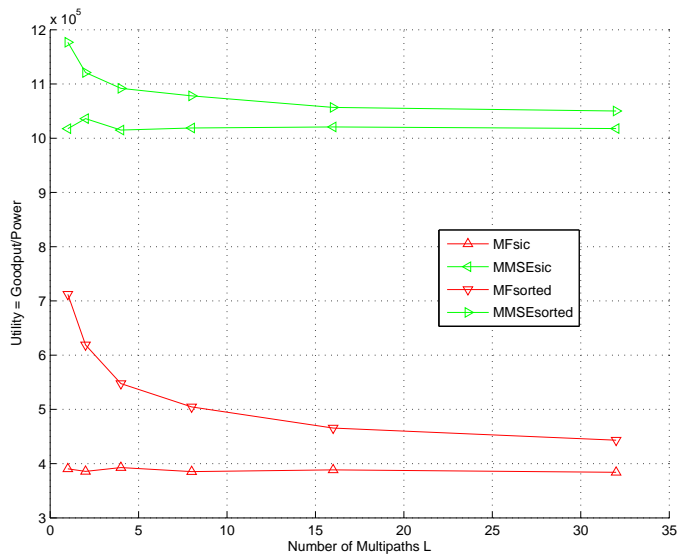}}
\caption{Simulation of utilities in the Nash equilibrium with SIC filter with and without optimal ordering, versus $L$.}
\label{fig:ordering}
\end{figure}

\section{Conclusion}
\label{sec:conclusion}

Using tools of random matrices, we have derived the equilibrium power
allocation in a game-theoretic framework applied to asymptotic CDMA
with cyclic prefix, under frequency-selective fading. Three receivers
are considered: matched filter, MMSE and optimum filter (given by
Shannon's capacity). In addition, distributed ordering mechanisms are
introduced and the successive interference cancellation variants of
the linear filters are studied. For each user, this power allocation
depends only on the total energy of the channel of the user under
consideration. For a frequency-flat channel, the power allocation
among users is dis-uniform, whereas when the number of multipaths
increases, the power allocation tends more and more to a uniform one.

\section{Appendix}

\subsection{Proof of Prop. \ref{prop:opt}}
\label{sec:proof:opt}

Notice that when $\sigma^2\rightarrow\infty$, $\Gamopt=0$, $\Gammmse=0$ and $\betammse(x)=\beta(x)=0$. Thus we only have to prove that the derivatives of either side
of \eqref{eq:optmmseeq} are equal.

Using $\rho(f,x)=P(x)\abs{h(f,x)}^2$, \eqref{eq:implicit:beta} can be rewritten
\begin{equation}
\beta(x)=\int_0^1\frac{\rho(f,x)df}{\sigma^2+\int_0^{\alpha}\frac{\rho(f,y)2dy}{1+\beta(y)}}.
\label{eq:implicit:betarho}
\end{equation}

From \eqref{eq:mrztheorem}, $\int_0^1 \rho(f,x) u(f,-\sigma^2)df$
satisfies the same implicit equation \eqref{eq:implicit:betarho} as $\beta(x)$
and thus
\begin{equation}
u(f,-\sigma^2) = \frac{1}{\int_0^\alpha\frac{\rho(f,y)dy}{1+\beta(y)}+\sigma^2}.
\label{eq:link}
\end{equation}
Using \eqref{eq:implicit:betarho} and \eqref{eq:link}, we can rewrite
\begin{align*}
\int_0^1 u&(f,-\sigma^2)df-\frac{1}{\sigma^2}\\ &=\int_0^1\frac{1}{\int_0^\alpha\frac{\rho(f,y)dy}{1+\beta(y)}+\sigma^2}df-\int_0^1\frac{1}{\sigma^2}df\\
&=\int_0^1\frac{-\int_0^\alpha\frac{\rho(f,x)}{1+\beta(x)}dx}{\sigma^2\left(\int_0^\alpha\frac{\rho(f,y)dy}{1+\beta(y)}+\sigma^2\right)}df\\
&=\int_0^\alpha\frac{\frac{-1}{\left(1+\beta(x)\right)}}{\sigma^2}\int_0^1\frac{\rho(f,x)df}{\int_0^\alpha\frac{\rho(f,y)dy}{1+\beta(y)}+\sigma^2} dx\\
&=-\int_0^\alpha\frac{\beta(x)}{\sigma^2\left(1+\beta(x)\right)}dx.
\end{align*}
Thus from \eqref{eq:scderived}
\begin{equation}
\frac{\partial\Gamopt}{\partial\sigma^2}=-\log_2(e)\int_0^\alpha\frac{\beta(x)}{\sigma^2\left(1+\beta(x)\right)}dx.
\label{eq:interestingequation}
\end{equation}

Differentiating \eqref{eq:gammmse} with respect to $\sigma^2$, we obtain
\begin{equation}
\frac{\partial\Gammmse}{\partial\sigma^2}=\log_2(e)\int_0^\alpha\frac{1}{1+\beta(x)}\frac{\partial\beta}{\partial{\sigma^2}}(x)dx.
\label{eq:derivemmse}
\end{equation}
Let $\pi(x)=\frac{1}{\sigma^2(1+\beta(x))}$.
From \eqref{eq:interestingequation} and \eqref{eq:derivemmse}, we obtain
\begin{multline}
\frac{\partial\Gamopt}{\partial\sigma^2}-\frac{\partial\Gammmse}{\partial\sigma^2}\\=-\log_2(e)\int_0^\alpha\left(\beta(x)+\sigma^2\frac{\partial\beta}{\partial{\sigma^2}}(x)\right)\pi(x)dx.
\label{eq:diff}
\end{multline}
From \eqref{eq:implicit:beta}, we have
\begin{align*}
&\int_0^\alpha \sigma^2\beta(x)\frac{\partial\pi}{\partial\sigma^2}(x)dx\\
&=\int_0^\alpha\int_0^1\frac{\sigma^2\rho(f,x)df}{\sigma^2+\int_0^\alpha\sigma^2\rho(f,y)\pi(y)dy}\frac{\partial\pi}{\partial\sigma^2}(x)dx\\
&=\int_0^1\frac{\int_0^\alpha\rho(f,x)\frac{\partial\pi}{\partial\sigma^2}(x)dx}{1+\int_0^\alpha\rho(f,y)\pi(y)dy}df\\
&=\frac{1}{\log_2(e)}\frac{\partial}{\partial\sigma^2}\int_0^1\log_2\left(1+\int_0^\alpha\rho(f,y)\pi(y)dy\right)df.
\end{align*}
Observing that
\begin{multline*}
\int_0^\alpha\left(\beta(x)+\sigma^2\frac{\partial\beta}{\partial{\sigma^2}}(x)\right)\pi(x)+\sigma^2\beta(x)\frac{\partial\pi}{\partial\sigma^2}(x)dx\\=\frac{\partial}{\partial\sigma^2}\int_0^\alpha\sigma^2\beta(x)\pi(x)dx
\end{multline*}
we obtain \eqref{eq:optmmseeq} from Prop. \ref{prop:opt}.

\subsection{Influence of Other Players' Strategies}
\label{sec:proof:large}

We want to prove that asymptotically, in the game
$\{S^K,\PA,(u_k)_{k\in S^K}\}$, the strategy of a single player does
not have any influence on the payoff of the other players. In other
words, for all $k\neq i\in S^K$, for all
$\p=(P_1,\dotsc,P_K)\in\PA^K$, for all $P'_i\in\PA$,
\begin{equation*}
\abs{u_k(\p)-u_k(P'_i,\p_{(-i)})}\to0,\ \textrm{as}\ N\to\infty.
\end{equation*}

Remember that $u_k = \frac{\Gami(\beta_k)}{P_k}$, and $\Gami$ is at least $C^2$. 
Let $(\beta_1,\dotsc,\beta_K)$ be the SINRs associated with the power allocation $\p$ and 
$(\beta'_1,\dotsc,\beta'_K)$ the SINRs associated with the power allocation $(P'_i,\p_{(-i)})$. 
Then a simple Taylor expansion of $\Gami$ in $\beta'_k$ gives
\begin{equation}
\Gami(\beta'_k)=\Gami(\beta_k)+(\beta'_k-\beta_k)\frac{\partial\Gami}{\partial\beta}(\beta_k)+o(\beta'_k-\beta_k).
\label{eq:dlGami}
\end{equation}
According to \eqref{eq:dlGami}, it is sufficient to show that 
\begin{equation}
\abs{\frac{\beta'_k-\beta_k}{P_k}}\to0,\ \textrm{as}\ N\to\infty. 
\label{eq:toprove}
\end{equation}

\paragraph{Matched Filter}
For the matched filter, the inequality is obtained directly from
\eqref{eq:finite:betamf}.  The denominator of \eqref{eq:finite:betamf}
is always greater than
$\frac{\sigma^2}{N}\sum_{n=1}^N\abs{\dd{k}{n}}^2$. Hence,
\begin{align*}
\abs{\frac{\beta'_k-\beta_k}{P_k}}&\leq\abs{\frac{P_k\frac{1}{N}(P'_i-P_i)\frac{1}{N}\sum_{n=1}^N\abs{\dd{i}{n}}^2\abs{\dd{k}{n}}^2}{P_k\sigma^4}}\\
&\leq\frac{\Pmax\hmax^2}{\sigma^4 N}.
\end{align*}

\paragraph{MMSE Filter}
For the MMSE filter, the inequality is obtained from
\eqref{eq:exactsinrmmse}, Lemma 1 from \cite{ieeeIT.tse00} and Lemma
2.1 from \cite{anlpro.bai07}, which we both reproduce below for
convenience.
\begin{lemma} \label{lemmabai} \cite{ieeeIT.tse00} 
Let $\HC$ be a $N\times N$ complex matrix with uniformely bounded
spectral radius for all $N$: $\sup_N(\abs{\HC})<\infty$. Let
$\ww=\frac{1}{\sqrt{N}}[w_1,\dotsc,w_N]^T$ where $\{w_i\}_{i=1\dotsc
N}$ are i.i.d.~complex random variables with zero mean, unit variance
and finite eighth moment. Then:
\begin{equation*}
\Expect{}{\abs{\ww^H\HC\ww-\frac{1}{N}\tr{\HC}}^4}\leq\frac{C}{N^2}
\end{equation*}
where $C$ is a constant that does not depend on $N$ or $\HC$. 
\end{lemma}
\begin{lemma} \label{lemmabai2} \cite{anlpro.bai07} Let $\sigma^2>0$, ${\bf A}$ and ${\bf B}$ $N\times N$ with ${\bf B}$ Hermitian nonnegative definite, and ${\bf q}\in\CC^N$. 
Then \begin{equation*}
\tr\left(\left(({\bf B}+\sigma^2\Id)^{-1}-({\bf B}+{\bf q}{\bf q}^H+\sigma^2\Id)^{-1}\right){\bf A}\right)\leq\frac{\norm{{\bf A}}}{\sigma^2}. 
\end{equation*}
\end{lemma}
In Lemma \ref{lemmabai2}, $\norm{{\bf A}}$ is the spectral norm of ${\bf A}$,
i.e., the square root of the largest singular value of ${\bf A}$.

From (7), we can write
\begin{gather*}
\beta_k = P_k \WWu{k}^H \HH_{k}^H \left(\Gu{k}\GuH{k}+\sigma^2\Id_N\right)^{-1} \HH_{k} \WWu{k},\\
\beta'_k = P_k \WWu{k}^H \HH_{k}^H\left({\Gu{k}}'{\GuH{k}}'+\sigma^2\Id_N\right)^{-1}\HH_{k} \WWu{k}
\end{gather*}
where
\begin{multline*}
{\Gu{k}}'{\GuH{k}}'=\Gu{k}\GuH{k}\\+(P'_i-P_i)(\HHu{i}\odot\WWu{i})(\HHu{i}\odot\WWu{i})^H.
\end{multline*}

A corollary of Lemma \ref{lemmabai} is that for
either matrix $\HC=\HH_{k}^H
\left(\Gu{k}\GuH{k}+\sigma^2\Id_N\right)^{-1} \HH_{k}$ or matrix
$\HC=\HH_{k}^H\left({\Gu{k}}'{\GuH{k}}'+\sigma^2\Id_N\right)^{-1}\HH_{k}$,
we obtain \cite{ieeeIT.tse00}
\begin{equation*}
\abs{\WWu{k}^H \HC \WWu{k}-\frac{1}{N}\tr{\HC}}\to0,\ \textrm{as}\ N\to\infty.
\end{equation*}

Matrix ${\bf B}=\Gu{k}\GuH{k}$ is Hermitian nonnegative definite, as for all $\ww\in\CC^N$, $\ww^H\Gu{k}\GuH{k}\ww=\norm{\Gu{k}\ww}^2 \geq 0$.
Diagonal matrix ${\bf A}=\HH_{k}\HH_{k}^H$ has spectral norm $\norm{\HH_{k}\HH_{k}^H}\leq\hmax^2$. 
Using Lemmas \ref{lemmabai} and \ref{lemmabai2}, as $N\to\infty$, we obtain
\begin{equation*}
\abs{\frac{\beta'_k-\beta_k}{P_k}}\to0,\ \textrm{as}\ N\to\infty.
\end{equation*} 

\paragraph{Optimum and Successive Interference Cancellation Filters}
The analog of the SINR derived for the optimum filter stems from the
MMSE filter with SIC. The SINR for SIC filters have similar
expressions with less interfering users appearing in the
denominator. Hence the result is immediate.

\subsection{Proof of Prop. \ref{prop:optflat}}
\label{sec:proof:optflat}

Given $\Gamoptstar$, we can use \eqref{eq:optmmseeq} to obtain a Nash equilibrium power allocation in the 
following way. We rewrite \eqref{eq:optmmseeq} assuming that the target SINR for the MMSE filter is $\betaplus$.  
\begin{multline}
\alpha\log_2\left(1+\betaplus\right)-\alpha\log_2(e)\frac{\betaplus}{1+\betaplus}\\+\log_2\left(1+\frac{1}{\sigma^2\left(1+\betaplus\right)}\int_0^\alpha P(y)\abs{h(y)}^2dy\right)\\ = \alpha\log_2\left(1+\betastar\right). 
\label{eq:optmmsebetaplus}
\end{multline}
In the left-hand side of \eqref{eq:optmmsebetaplus}, $P(y)$ is given by a  MMSE power allocation similar to the one given by \eqref{eq:flatminpower}. 
Hence, the term $\int_0^\alpha P(y)\abs{h(y)}^2dy$ in \eqref{eq:optmmsebetaplus} does not depend on the actual realizations of the channels. 
Replacing $\betastar$ by  $\betaplus$ in \eqref{eq:flatdiscrete}, we obtain that  
$\int_0^\alpha P(y)\abs{h(y)}^2dy=\frac{\alpha\sigma^2\betaplus}{1-\alpha\frac{\betaplus}{1+\betaplus}}$, which gives us \eqref{eq:optmmsebetaplusprop}. 
Replacing $\betastar$ by  $\betaplus$ in \eqref{eq:flatminpower}, we obtain the power allocation \eqref{eq:optflatminpower}. 

\subsection{Expectation of the random variable \eqref{eq:converge}}
\label{sec:proof:converge}

For each user $j$, there are $L>1$ paths with respective attenuations 
$h_{\ell}\left(\frac{j}{N}\right),\ \ell=1,\dotsc,L$, which are i.i.d.~complex random variables
with mean zero and even distributions of the real and imaginary parts. The Fourier transform of those attenuations is
$\dd{j}{n}=h\left(\frac{n}{N},\frac{j}{N}\right)=\sum_{\ell=1}^{L}h_{\ell}\left(\frac{j}{N}\right)e^{-2\pi i \frac{n}{N} (\ell-1)}$.
The total energy of the paths is $\E_j = \sum_{\ell=1}^L \abs{h_{\ell}\left(\frac{j}{N}\right)}^2$.

We want to show that the expectation of the random variable
$\frac{1}{K}\sum_{j=1}^K\frac{\abs{\dd{j}{n}}^2}{\E_j}$ is equal to 1.
By expanding the expression of $\dd{j}{n}$, this is equivalent to showing that the expectation of the random variable
\begin{equation*}
\frac{h_{\ell}\left(\frac{j}{N}\right)h_{\ell'}\left(\frac{j'}{N}\right)}{\E_j}
\end{equation*}
is equal to 0. Denoting by $p(\cdot)$ the distribution of $h_{\ell}=h_{\ell}\left(\frac{j}{N}\right)$, this expectation is equal to the $L$-dimensional integral of
\begin{equation*}
\frac{h_{\ell}h_{\ell'}}{\abs{h_{\ell}}^2+\abs{h_{\ell'}}^2+\sum_{k\neq\ell,\ell'}\abs{h_{k}}^2}p\left(h_{\ell}\right)p\left(h_{\ell'}\right)\prod_{k\neq\ell,\ell'}p\left(h_{k}\right)
\end{equation*}
which is an odd function of $h_{\ell}$. Its integral is therefore 0, which proves the desired result. 

\subsection{Proof of \eqref{eq:pamfsicflat} and \eqref{eq:pammsesicflat}}
\label{sec:PAsicflat}

Denote $m_k = P_{K-k}\abs{h_{K-k}}$. From \eqref{eq:finite:pamfsic}, with flat fading, 
the sequence $\{m_k\}_{k\in S^K}$ satisfies $m_0=\betastar\sigma^2$ and $m_{k+1}=\betastar\sigma^2+\frac{\betastar}{N}\sum_{j=0}^k m_j$.
Using the fact that $\sum_{i=j}^k\binom{i}{j}=\binom{k+1}{j+1}$, it is immediate to prove by recurrence that 
\begin{equation*}
m_{k}=\betastar\sigma^2\sum_{j=0}^{k}\binom{k}{j}\frac{1}{N^j}{\betastar}^j=\betastar\sigma^2\left(1+\frac{1}{N}\betastar\right)^k.
\end{equation*}
Hence formula \eqref{eq:pamfsicflat}. The demonstration is exactly
similar for \eqref{eq:pammsesicflat} from the recursion $m_0=\betastar\sigma^2$ and
$m_{k+1}=\betastar\sigma^2+\frac{\betastar}{\left(1+\betastar\right)N}\sum_{j=0}^k
m_j$.

\subsection{Optimal Ordering of Users}
\label{sec:ordering}

We determine the ordering that makes use of the least total power for
equilibrium PA \eqref{eq:pamfsicflat} (the case is similar for
\eqref{eq:pammsesicflat}, \eqref{eq:pamfsic} and \eqref{eq:pammsesic}).
Let the ordering of the users be such as
$\abs{h_1}^2<\dotsb<\abs{h_K}^2$. Let $\pi$ be any permutation
of $\{1,\dotsc,K\}$.  Let
$a_{ij}=\left(1+\frac{1}{N}\betastar\right)^{K-i}-\left(1+\frac{1}{N}\betastar\right)^{K-j}$.

Then showing that the optimal ordering is such as $\abs{h_1}^2<\dotsb<\abs{h_K}^2$ is equivalent to showing that for any $\pi$
\begin{equation}
\sum_{k=1}^K \frac{1}{\abs{h_k}^2} a_{k\pi(k)}> 0.
\label{eq:equivalence}
\end{equation}

Consider first a cyclic permutation. By the definition of $a_{ij}$, the sum of the
$a_{k\pi(k)}$ is equal to zero:
$\sum_{k=1}^K a_{k\pi(k)}=0$. The first coefficient $a_{1\pi(1)}$ is
positive. It is affected coefficient $\frac{1}{\abs{h_1}^2}$, which is
the greatest coefficient in the sum in \eqref{eq:equivalence}. Hence
the sum in \eqref{eq:equivalence} is positive in this case.

Permutation $\pi$ can be decomposed as a product of disjoint
permutation cycles. Each cycle determines a subset of indexes $k$,
these subsets form a partition of $\{1,\dotsc,K\}$. With a similar
reasoning as precedently, replacing index $1$ with the smallest index
in the cycle, the sum over the indexes $k$ pertaining to a cycle of
$\frac{1}{\abs{h_k}^2} a_{k\pi(k)}$ is positive. Hence the global sum of
\eqref{eq:equivalence} is also positive.

It can be proven in a similar way that the same ordering maximizes the
sum of inverse powers of the users.

\section{Acknowledgements}

The authors would like to thank Prof. J. Silverstein for pointing us to reference \cite{anlpro.bai07}.



\end{document}